\documentclass[journal]{IEEEtran}%
\pdfoutput=1
\usepackage{amsfonts}
\usepackage{amsmath}
\usepackage{amssymb}
\usepackage{graphicx}%
\providecommand{\U}[1]{\protect\rule{.1in}{.1in}}
\newtheorem{theorem}{Theorem}

\newtheorem{lemma}[theorem]{Lemma}

\newtheorem{proposition}[theorem]{Proposition}

\newenvironment{proof}[1][Proof]{\noindent\textbf{#1.} }{\ \rule{0.5em}{0.5em}}
\begin{document}

\title{Polar codes for degradable quantum channels}
\author{Mark M. Wilde and Saikat Guha\thanks{Mark M. Wilde is with the School of
Computer Science, McGill University, Montreal, Quebec H3A 2A7, Canada. Saikat
Guha is with the Quantum Information Processing Group,
Raytheon BBN Technologies, Cambridge, Massachusetts, USA 02138. (E-mail:
mark.wilde@mcgill.ca; sguha@bbn.com)}}
\maketitle

\begin{abstract}
Channel polarization is a phenomenon in which a particular recursive encoding
induces a set of synthesized channels from many instances of a memoryless
channel, such that a fraction of the synthesized channels becomes near perfect
for data transmission and the other fraction becomes near useless for this
task. Mahdavifar and Vardy have recently exploited this phenomenon to
construct codes that achieve the symmetric private capacity for private data
transmission over a degraded wiretap channel. In the current paper, we build
on their work and demonstrate how to construct quantum wiretap polar codes
that achieve the symmetric private capacity of a degraded quantum wiretap
channel with a classical eavesdropper. Due to the Schumacher-Westmoreland
correspondence between quantum privacy and quantum coherence, we can construct
quantum polar codes by operating these quantum wiretap polar codes in
superposition, much like Devetak's technique for demonstrating the
achievability of the coherent information rate for quantum data transmission.
Our scheme achieves the symmetric coherent information rate for quantum
channels that are degradable with a classical environment. This condition on
the environment may seem restrictive, but we show that many quantum channels
satisfy this criterion, including amplitude damping channels, photon-detected
jump channels, dephasing channels, erasure channels, and cloning channels. Our
quantum polar coding scheme has the desirable properties of being
channel-adapted and symmetric capacity-achieving along with having an
efficient encoder, but we have not demonstrated that the decoding is
efficient. Also, the scheme may require entanglement assistance, but we show
that the rate of entanglement consumption vanishes in the limit of large
blocklength if the channel is degradable with classical environment.

\end{abstract}

\section{Introduction}

In a seminal paper on quantum error correction, Shor set out the
\textquotedblleft goal of [defining] the quantum analog of the Shannon
capacity \cite{S48}\ for a quantum channel, and [finding] encoding schemes
which approach this capacity\textquotedblright~\cite{S95}. At the time, it was
not really clear how to define the quantum capacity of a quantum channel, but
Shor's quantum error correction code \cite{S95} gave some clues for
constructing more general encoding schemes. Subsequently, several authors
contributed increasingly sophisticated quantum error correction codes
\cite{CS96,S96,G97}\ and others established a good definition of and upper
bounds on the quantum capacity of a channel
\cite{PhysRevA.54.2614,PhysRevA.54.2629,BKN98,BNS98}, culminating in some
high-performing quantum error-correction codes \cite{MMM04,L08,PTO09,KHIS10}%
\ and random-coding based schemes for achieving the coherent information rate
\cite{PhysRevA.54.2629} of a quantum
channel~\cite{PhysRevA.55.1613,capacity2002shor,D03}. For some channels known
as degradable quantum channels \cite{cmp2005dev}, in which the channel to the
environment is noisier than the channel to the intended receiver, the
random-coding based schemes from
Refs.~\cite{PhysRevA.55.1613,capacity2002shor,D03} achieve their quantum
capacity, due to the particular structure of these channels.

In spite of the astounding progress in both quantum error
correction~\cite{DNM09}\ and quantum Shannon theory~\cite{W11}, none of the
high performance codes constructed to date are provably capacity achieving,
and none of the aforementioned schemes that achieve the capacity are explicit
(the proofs instead exploit randomness to establish the existence of a code).
Among the schemes that achieve the quantum capacity, perhaps
Devetak's~\cite{D03} provides the most clear recipe to a quantum code designer
interested in constructing a capacity-achieving quantum code. His proof takes
a cue from a certain security proof of quantum key
distribution~\cite{PhysRevLett.85.441} and the Schumacher-Westmoreland
correspondence between quantum privacy and quantum
coherence~\cite{PhysRevLett.80.5695},\ by first establishing the existence of
codes that achieve the private capacity of a quantum wiretap channel and then
demonstrating how to operate such a code in superposition so that it achieves
the quantum capacity. It is also clear that the structure of his codes bears
some similarities with Calderbank-Shor-Steane codes~\cite{CS96,S96}.

Along with the above developments, there have been impressive breakthroughs in
classical coding theory and information theory \cite{FC07}, one of which is
Arikan's recent work on polar codes \cite{A09}. Arikan's polar codes exploit a
phenomenon known as channel polarization, in which a particular recursive
encoding induces a set of synthesized channels from the original memoryless
noisy channels. The synthesized channels are such that a fraction of them are
perfect for data transmission, while the other fraction are completely
useless, and the fraction that are perfect is equal to the symmetric capacity
of the original channel. The codes are channel adapted, in the sense that
Arikan's \textquotedblleft polar coding rule\textquotedblright\ establishes
through which of the synthesized channels the sender should transmit data, and
this polar coding rule depends on the particular channel being used. The codes
are near explicit and have the desirable property that both the encoding and
decoding are efficient (the complexity of each is $O\left(  N\log N\right)  $
where $N$ is the blocklength of the code).

Arikan's work might make us wonder whether it would be possible to construct
polar codes for transmitting quantum data over general quantum channels, and
the development in the classical world most relevant for this task is due to
Madhavifar and Vardy \cite{MV10}. There, they established that a modification
of Arikan's original polar coding scheme can achieve the symmetric private
capacity of a degraded classical wiretap channel. (In order to make this
statement, the sender and receiver actually require access to a small amount
of secret key, but the rate of secret key needed vanishes when the code's
blocklength becomes large.) Thus, with the Madhavifar-Vardy scheme for polar
coding over classical degraded wiretap channels \cite{MV10}, the Devetak
scheme for operating a quantum wiretap code in superposition \cite{D03}, and our
recent work on polar codes for transmitting classical data over quantum
channels \cite{WG11}, it should be evident that one could put these pieces
together in order to construct polar codes for transmitting quantum data over
degradable quantum channels.

In this paper, we pursue this direction by constructing polar codes that
achieve a symmetric capacity for transmitting quantum data over particular
degradable quantum channels. These channels should satisfy the property that
encoding classical data in some orthonormal basis at their input leads to
commuting states for the environment (essentially, the environment becomes
classical), and we clarify later why this is important in our construction.
Many degradable channels fall into this class, including amplitude damping
channels \cite{GF05}, photon-detected jump channels \cite{ABCDGM01}, erasure
channels~\cite{GBP97}, dephasing channels~\cite{D03}, and cloning channels
\cite{B11}\ (channels induced by universal cloning machines~\cite{GM97,BH98}).
These noisy channels occur naturally in physical processes, with amplitude
damping modeling photon loss or spontaneous emission, the photon-detected jump
channel modeling the spontaneous decay of atoms with a detected photon
emission \cite{ABCDGM01}, the erasure channel being a different model for
photon loss \cite{GKLVD03}, the dephasing channel modeling random phase noise
in superconducting systems~\cite{BDKS08}, and the cloning channel modeling
stimulated emission from an atom~\cite{MH82,SWZ00,LSHB02}. Our codes are
symmetric capacity achieving for all of the above channels, and this follows
from analyzing a quantum polar coding rule for these channels.

We summarize briefly how the construction works. First, we consider a quantum
wiretap channel with one classical input and two quantum outputs, one for the
legitimate receiver (Bob)\ and the other for the wiretapper (Eve). We
demonstrate that these channels polarize in four different ways, based on
whether the channels are good or bad for the receiver or the wiretapper. In
order to have strong security, we must guarantee that the bad channels for the
wiretapper are in fact \textquotedblleft really bad\textquotedblright\ in a
precise sense, and it is for this reason that we consider channels with
classical environment. (Interestingly, we know of quantum channels where it is
not clear to us how to ensure that they become \textquotedblleft really
bad,\textquotedblright\ and we prove some results in this direction in
Appendix~\ref{app:pure-state-combining}.) The resulting coding scheme is to
send the information bits through the channels which are good for Bob and bad
for Eve, \textquotedblleft frozen\textquotedblright\ bits through the channels
that are bad for both, half of a secret key through channels which are bad for
Bob but good for Eve, and randomized bits through the channels good for both.
By an analysis similar to that of Madhavifar and Vardy \cite{MV10}, we can
demonstrate that this scheme achieves the symmetric private capacity of a
degraded quantum wiretap channel with classical environment, while the rate of
secret key required vanishes in the limit of large blocklength.

The main idea for constructing quantum polar codes for degradable quantum
channels with classical environment is just to operate the quantum wiretap
code in superposition and exploit Arikan's encoding with CNOT gates with
respect to some orthonormal basis. This amounts to sending information qubits
through the channels good for the receiver Bob and bad for the
\textit{environment} Eve, frozen ancilla qubits through the channels that are
bad for both, half of shared entanglement through the channels that are bad
for Bob but good for Eve, and superposed ancilla qubits in the state
$\left\vert +\right\rangle \equiv\left(  \left\vert 0\right\rangle +\left\vert
1\right\rangle \right)  /\sqrt{2}$ through the channels that are good for
both. The resulting quantum polar codes are entanglement-assisted
\cite{BDH06}, but we can prove that the entanglement consumption rate required
vanishes in the limit of large blocklength (in this regard, the codes here are
similar to Hsieh \textit{et al}.'s recent ones~\cite{HYH11}). Operating the
quantum successive cancellation decoder from Ref.~\cite{WG11}\ in a coherent
fashion, followed by controlled \textquotedblleft decoupling
unitaries\textquotedblright\ allows us to exploit the properties of the
quantum wiretap polar code in order to prove that the above quantum polar code
performs well (we note that this decoder is similar to Devetak's \cite{D03}).
The resulting quantum polar codes achieve the symmetric quantum capacity of
degradable channels with classical environment with an encoding circuit that
has complexity $O\left(  N\log N\right)  $. The decoding unfortunately remains
inefficient, but further efforts may lead to an efficient realization of a decoder.

Recently, Renes \textit{et al}.~have independently constructed quantum polar
codes that have both an efficient encoding and decoding, though they achieve
the coherent information rate only for Pauli channels \cite{RDR11}. We should
clarify the ways in which their scheme is different from ours. First, they
restrict their construction to Pauli channels because they are considering the
effective classical channels induced in the amplitude (Pauli-$Z$) and phase
(Pauli-$X$) bases (not all channels, including some of the ones mentioned
above, induce classical channels in complementary bases). As a result, they
can directly import Arikan's ideas because they are dealing with classical
channels in complementary bases. An additional bonus is that they obtain an
efficient decoder as well as an efficient encoder, essentially because their
decoder is Arikan's successive cancellation decoder implemented as an
efficient unitary operation. Their codes require the assistance of shared
entanglement, but there are some channels for which they can prove that it is
not required.

Our scheme is different from theirs in several ways. First, we have a quantum
polar coding rule that is adapted to a given quantum channel. In particular,
the rule used for determining the good or bad channels is based on a quantum
parameter (fidelity). Also, we demonstrate polarization in terms of two
quantum parameters, the fidelity and the Holevo information, by building on
our earlier results in Ref.~\cite{WG11}. It is for this reason that our scheme
is symmetric capacity-achieving for a wide variety of quantum channels. Also,
the first part of our decoder is a coherent version of the quantum successive
cancellation decoder from Ref.~\cite{WG11}, rather than one that is based
directly on the classical decoder. Since we have not proven that the quantum
successive cancellation decoder from Ref.~\cite{WG11} has an efficient
implementation, the coherent version of it in this work is certainly not
efficient. In spite of the inefficiency of the quantum successive cancellation
decoder, this decoder is needed in order to achieve the symmetric quantum
capacity of the channels considered here. Finally, all the channels we
consider here only require a vanishing rate of entanglement assistance, due to
an argument similar to Proposition~22 of Madhavifar and Vardy~\cite{MV10} and
the fact that they are degradable with a classical environment.

We structure this paper as follows. Section~\ref{sec:cq-review}\ reviews our
work from Ref.~\cite{WG11}\ on polar codes for classical-quantum channels. We
present in Section~\ref{sec:wiretap}\ our scheme for private communication
over quantum wiretap channels that are degraded with a classical environment.
Finally, we demonstrate how to construct quantum polar codes from quantum
wiretap polar codes in Section~\ref{sec:quantum-polar}. The last section
concludes with a summary and some open questions.

\section{Review of polar codes for classical-quantum channels}

\label{sec:cq-review}We begin by providing a brief review of polar codes
constructed for classical-quantum channels~\cite{WG11}. There, we considered
channels with binary classical inputs and quantum outputs of the form:%
\[
W:x\rightarrow\rho_{x},
\]
where $W$ denotes the channel, $x\in\left\{  0,1\right\}  $, and $\rho_{x}$ is
a density operator. The relevant parameters that determine channel performance
are the fidelity $F\left(  W\right)  \equiv\left\Vert \sqrt{\rho_{0}}%
\sqrt{\rho_{1}}\right\Vert _{1}^{2}$ and the symmetric Holevo information
$I\left(  W\right)  \equiv H\left(  \rho\right)  -1/2(H\left(  \rho
_{0}\right)  +H\left(  \rho_{1}\right)  )$ with $\rho=1/2\left(  \rho_{0}%
+\rho_{1}\right)  $ and the von Neumann entropy $H\left(  \sigma\right)
\equiv-$Tr$\left\{  \sigma\log_{2}\sigma\right\}  $. Channels with $F\left(
W\right)  \leq\epsilon$ are nearly noiseless and those with $F\left(
W\right)  \geq1-\epsilon$ are near to being completely useless. The fidelity
generalizes the Bhattacharya distance $Z$ \cite{A09}\ in the sense that
$\sqrt{F\left(  W\right)  }=Z\left(  W\right)  $ if the two states $\rho_{0}$
and $\rho_{1}$ commute (i.e., if the channel is classical).

In coding classical information for the above channel, we consider $N=2^{n}$
copies of $W$, such that the resulting channel is of the form%
\[
x^{N}\equiv x_{1}\cdots x_{N}\rightarrow\rho_{x^{N}}\equiv\rho_{x_{1}}%
\otimes\cdots\otimes\rho_{x_{N}},
\]
where $x^{N}$ is the length $N$ input and $\rho_{x^{N}}$ is the output state.
We can extend Arikan's idea of channel combining to this classical-quantum
channel, by considering the channels induced by a transformation on an input
bit (row) vector $u^{N}$:%
\[
u^{N}\rightarrow\rho_{u^{N}G_{N}},
\]
where $G_{N}=B_{N}F^{\otimes n}$, with $B_{N}$ being a permutation matrix that
reverses the order of the bits and%
\[
F=%
\begin{bmatrix}
1 & 0\\
1 & 1
\end{bmatrix}
.
\]
This classical encoding is equivalent to a network of classical CNOT\ gates
and permutation operations that can be implemented with complexity $O\left(
N\log N\right)  $ (see Figures~1 and 2 of Ref.~\cite{WG11} or Figures~1, 2,
and 3 of Ref.~\cite{A09}). We can also define the split channels from the
above combined channels as%
\begin{equation}
W_{N}^{\left(  i\right)  }:u_{i}\rightarrow\rho_{\left(  i\right)  ,u_{i}%
}^{U_{1}^{i-1}B^{N}}, \label{eq:split-channels}%
\end{equation}
where%
\begin{align}
\rho_{\left(  i\right)  ,u_{i}}^{U_{1}^{i-1}B^{N}}  &  \equiv\sum_{u_{1}%
^{i-1}}\frac{1}{2^{i-1}}\left\vert u_{1}^{i-1}\right\rangle \left\langle
u_{1}^{i-1}\right\vert ^{U_{1}^{i-1}}\otimes\overline{\rho}_{u_{1}^{i}}%
^{B^{N}},\\
\overline{\rho}_{u_{1}^{i}}^{B^{N}}  &  \equiv\sum_{u_{i+1}^{N}}\frac
{1}{2^{N-i}}\rho_{u^{N}G_{N}}^{B^{N}}. \label{eq:averaged-cond-states}%
\end{align}
The interpretation of this channel is that it is the one \textquotedblleft
seen\textquotedblright\ by the bit $u_{i}$ if all of the previous bits
$u_{1}^{i-1}$ are available and if we consider all the future bits
$u_{i+1}^{N}$ as randomized. This motivates the development of a quantum
successive cancellation decoder \cite{WG11}\ that attempts to distinguish
$u_{i}=0$ from $u_{i}=1$ by adaptively exploiting the results of previous
measurements and Helstrom-Holevo measurements \cite{H69,Hol72}\ for each bit decision.

Arikan's polar coding rule is to divide the channels into \textquotedblleft
good\textquotedblright\ ones and \textquotedblleft bad\textquotedblright%
\ ones. Let $\left[  N\right]  \equiv\left\{  1,\ldots,N\right\}  $ and
$\beta$ be a real such that $0<\beta<1/2$. The polar coding rule divides the
channels as follows:%
\begin{align}
\mathcal{G}_{N}\left(  W,\beta\right)   &  \equiv\left\{  i\in\left[
N\right]  :\sqrt{F(W_{N}^{\left(  i\right)  })}<2^{-N^{\beta}}\right\}
,\label{eq:good-channels}\\
\mathcal{B}_{N}\left(  W,\beta\right)   &  \equiv\left\{  i\in\left[
N\right]  :\sqrt{F(W_{N}^{\left(  i\right)  })}\geq2^{-N^{\beta}}\right\}  ,
\label{eq:bad-channels}%
\end{align}
so that the channels in $\mathcal{G}_{N}\left(  W,\beta\right)  $ are the good
ones and those in $\mathcal{B}_{N}\left(  W,\beta\right)  $ are the bad ones.
Observe that the quantum polar coding rule involves the quantum channel
parameter $F$, rather than a classical one such as the Bhattacharya distance.

The following theorem is helpful in determining what fraction of the channels
become good or bad \cite{AT09}:

\begin{theorem}
[Convergence Rate]\label{thm:conv-rate}Let $\left\{  X_{n}:n\geq0\right\}  $
be a random process with $0\leq X_{n}\leq1$ and satisfying%
\begin{align}
X_{n+1}  &  \leq qX_{n}\ \ \ \ \ \text{w.p.\ \ \ }%
1/2,\label{eq:general-process-relations-4-conv}\\
X_{n+1}  &  =X_{n}^{2}\ \ \ \ \ \text{w.p.\ \ \ }1/2,
\label{eq:general-process-relations-4-conv-2}%
\end{align}
where $q$ is some positive constant. Let $X_{\infty}=\lim_{n\rightarrow\infty
}X_{n}$ exist almost surely with $\Pr\left\{  X_{\infty}=0\right\}
=P_{\infty}$. Then for any $\beta<1/2$,%
\[
\lim_{n\rightarrow\infty}\Pr\left\{  X_{n}<2^{-2^{n\beta}}\right\}
=P_{\infty},
\]
and for any $\beta>1/2$,%
\[
\lim_{n\rightarrow\infty}\Pr\left\{  X_{n}<2^{-2^{n\beta}}\right\}  =0.
\]

\end{theorem}

One can then consider the channel combining and splitting mentioned above as a
random birth process in which a channel $W_{n+1}$ is constructed from two
copies of a previous one $W_{n}$ according to the rules in Section~4 of
Ref.~\cite{WG11}. One can then consider the process $\left\{  F_{n}%
:n\geq0\right\}  \equiv\left\{  \sqrt{F\left(  W_{n}\right)  }:n\geq0\right\}
$ and prove that it is a bounded super-martingale by exploiting the
relationships given in Proposition~10 of Ref.~\cite{WG11}. From the
convergence properties of martingales, one can then conclude that $F_{\infty}$
converges almost surely to a value in $\left\{  0,1\right\}  $, and the
probability that it equals zero is equal to the symmetric Holevo
information$~I\left(  W\right)  $. Furthermore, since the process $F_{n}$
satisfies the relations in (\ref{eq:general-process-relations-4-conv}%
-\ref{eq:general-process-relations-4-conv-2}), the following proposition on
the convergence rate of polarization holds:

\begin{theorem}
\label{thm:fraction-good}Given a binary input classical-quantum channel $W$
and any $\beta<1/2$,%
\[
\lim_{n\rightarrow\infty}\Pr\left\{  F_{n}<2^{-2^{n\beta}}\right\}  =I\left(
W\right)  .
\]

\end{theorem}

One of the important advances in Ref.~\cite{WG11} was to establish that a
quantum successive cancellation decoder performs well for polar coding over
classical-quantum channels. In this case, the decoder is some positive
operator-valued measure (POVM)\ $\left\{  \Lambda_{u_{\mathcal{A}}}\right\}  $
that attempts to decode the information bits $u_{\mathcal{A}}$ reliably. In
particular, we showed the following bound on the performance of such a decoder
(by exploiting Sen's \textquotedblleft non-commutative union
bound\textquotedblright~\cite{S11}):%
\[
\Pr\{\widehat{U}_{\mathcal{A}}\neq U_{\mathcal{A}}\}\leq2\sqrt{\sum
_{i\in\mathcal{A}}\frac{1}{2}\sqrt{F(W_{N}^{\left(  i\right)  })}},
\]
under the assumption that the sender chooses the information bits
$U_{\mathcal{A}}$ according to a uniform distribution. Thus, by choosing the
channels over which the sender transmits the information bits to be in
$\mathcal{G}_{N}\left(  W,\beta\right)  $ and those over which she transmits
agreed upon frozen bits to be in $\mathcal{B}_{N}\left(  W,\beta\right)  $, we
obtain the following bound on the probability of decoding error:%
\[
\Pr\{\widehat{U}_{\mathcal{A}}\neq U_{\mathcal{A}}\}=o(2^{-\frac{1}{2}%
N^{\beta}}).
\]
This completes the specification of a polar code for classical-quantum channels.

We end this section by stating a lemma that will prove useful for us:

\begin{lemma}
\label{lem:degrade}Let $W$ and $W^{\ast}$ both be binary-input
classical-quantum channels, such that $W^{\ast}$ is a degraded version of $W$,
in the sense that%
\[
W^{\ast}\left(  x\right)  =\mathcal{D}\left(  W\left(  x\right)  \right)  ,
\]
where $x$ is the classical input to the channels and $\mathcal{D}$ is some
degrading quantum channel from $W$ to $W^{\ast}$. Let $W_{N}^{\left(
1\right)  }$, \ldots, $W_{N}^{\left(  N\right)  }$ and $W_{N}^{\ast\left(
1\right)  }$, \ldots, $W_{N}^{\ast\left(  N\right)  }$ denote the
corresponding synthesized channels from channel combining and splitting. Then
$W_{N}^{\ast\left(  i\right)  }$ is degraded with respect to $W_{N}^{\left(
i\right)  }$ for all $i\in\left[  N\right]  $ and furthermore, we have that
$I(W_{N}^{\left(  i\right)  })\geq I(W_{N}^{\ast\left(  i\right)  })$ and
$F(W_{N}^{\left(  i\right)  })\leq F(W_{N}^{\ast\left(  i\right)  })$.
\end{lemma}

\begin{proof}
This lemma follows straightforwardly from the definition in
(\ref{eq:split-channels}), and the fact that quantum mutual information and
fidelity are monotone under quantum processing with the degrading map
$\mathcal{D}$~\cite{W11}.
\end{proof}

We can then observe from the above lemma and (\ref{eq:good-channels}) that if
$W^{\ast}$ is degraded with respect to $W$, the good channels for $W^{\ast}$
are a subset of those that are good for $W$: $\mathcal{G}_{N}\left(  W^{\ast
},\beta\right)  \subseteq\mathcal{G}_{N}\left(  W,\beta\right)  $. Similarly,
the following relationship holds as well:\ $\mathcal{B}_{N}\left(
W,\beta\right)  \subseteq\mathcal{B}_{N}\left(  W^{\ast},\beta\right)  $.

\section{Quantum wiretap polar codes}

\label{sec:wiretap}We now discuss how to construct polar codes that achieve
the symmetric private information rate for a quantum wiretap channel. The
results in this section build upon those of Mahdavifar and Vardy in
Ref.~\cite{MV10}.

The model for a binary-input quantum wiretap channel is as follows:%
\[
x\rightarrow\rho_{x}^{BE},
\]
where $x\in\left\{  0,1\right\}  $ and $\rho_{x}^{BE}$ is a density operator
on a tensor product Hilbert space $BE$. The legitimate receiver Bob has access
to the system $B$ and the eavesdropper Eve has access to the system $E$. Thus,
Bob's density operator is%
\[
\rho_{x}^{B}=\text{Tr}_{E}\left\{  \rho_{x}^{BE}\right\}  ,
\]
and Eve's density operator is%
\[
\rho_{x}^{E}=\text{Tr}_{B}\left\{  \rho_{x}^{BE}\right\}  .
\]
The quantum wiretap channel is degraded if there exists some quantum channel
$\mathcal{D}$ such that the following condition holds for all $x$:%
\[
\rho_{x}^{E}=\mathcal{D}\left(  \rho_{x}^{B}\right)  .
\]
Let $W$ denote the channel to Bob:%
\begin{equation}
W:x\rightarrow\rho_{x}^{B}, \label{eq:channel-to-Bob}%
\end{equation}
and let $W^{\ast}$ denote the channel to Eve:%
\begin{equation}
W^{\ast}:x\rightarrow\rho_{x}^{E}. \label{eq:channel-to-Eve}%
\end{equation}

In order to make a statement about the strong security of a quantum wiretap
polar code, we need to ensure that the channels over which the sender is
transmitting information bits to Bob should be \textquotedblleft really
bad\textquotedblright\ for Eve. That is, it is not sufficient for the channels
to satisfy (\ref{eq:bad-channels}), but they should be divided as to whether
they are poor for Eve according to the following stronger criterion:%
\[
\mathcal{P}\left(  W^{\ast},\beta\right)  \equiv\left\{  i\in\left[  N\right]
:\sqrt{F(W_{N}^{\ast\left(  i\right)  })}>1-2^{-N^{\beta}}\right\}  .
\]
Dividing the channels for Eve in this way makes it nearly impossible for her
to determine whether the sender transmits a zero or one through these channels
in the limit where $N$ becomes large. In what follows, we say that the
channels in $\mathcal{P}\left(  W^{\ast},\beta\right)  $ are \textquotedblleft
bad\textquotedblright\ for Eve while those in $\mathcal{P}^{c}\left(  W^{\ast
},\beta\right)  \equiv\left[  N\right]  \setminus\mathcal{P}\left(  W^{\ast
},\beta\right)  $ are \textquotedblleft good\textquotedblright\ for Eve.

It is again important for us to know what fraction of the channels
$W_{N}^{\ast\left(  i\right)  }$ become bad for Eve in order to establish that
the quantum wiretap polar codes are symmetric capacity-achieving---i.e., it
would be good to have another theorem similar to
Theorem~\ref{thm:fraction-good} for this case. In order to have such a
theorem, we would require a birth process that obeys the properties in
(\ref{eq:general-process-relations-4-conv}%
-\ref{eq:general-process-relations-4-conv-2}). Fortunately, in the case that
$\rho_{0}^{E}$ and $\rho_{1}^{E}$ commute, we have the following proposition:

\begin{proposition}
\label{prop:poor-channel-fraction}Suppose that the states $\rho_{0}^{E}$ and
$\rho_{1}^{E}$ for the binary-input classical-quantum channel $W^{\ast}$
commute. Then for any $\beta<1/2$,%
\[
\lim_{n\rightarrow\infty}\Pr\left\{  F_{n}^{\ast}>1-2^{-2^{n\beta}}\right\}
=1-I\left(  W^{\ast}\right)  ,
\]
where $F_{n}^{\ast}$ is the process $\left\{  F_{n}^{\ast}:n\geq0\right\}
\equiv\{\sqrt{F(W_{n}^{\ast})}:n\geq0\}$.
\end{proposition}

\begin{proof}
The proof proceeds along similar lines as Theorem~3.15 in Ref.~\cite{K09}.
Since the states $\rho_{0}^{E}$ and $\rho_{1}^{E}$ commute, they are
effectively classical, and the fidelities in $F_{n}^{\ast}$ reduce to the
classical Bhattacharya parameters. It is then possible to show that this
process satisfies%
\begin{align}
F_{n+1}^{\ast} &  \geq F_{n}^{\ast}\sqrt{2-\left(  F_{n}^{\ast}\right)  ^{2}%
}\ \ \ \text{w.p.\ \ \ }1/2,\label{eq:critical-inequality}\\
F_{n+1}^{\ast} &  =\left(  F_{n}^{\ast}\right)  ^{2}%
\ \ \ \ \ \ \ \ \ \ \ \ \ \ \ \text{w.p.\ \ \ }1/2.
\end{align}
The first relation follows from Lemma~3.16 in Ref.~\cite{K09}, and the second
follows from Lemma~2.16 in Ref.~\cite{K09}. We can then rewrite the above
conditions as follows:%
\begin{align}
1-\left(  F_{n+1}^{\ast}\right)  ^{2} &  \leq\left(  1-\left(  F_{n}^{\ast
}\right)  ^{2}\right)  ^{2}\ \ \ \ \ \ \ \ \ \ \text{w.p.\ \ \ }1/2,\\
1-\left(  F_{n+1}^{\ast}\right)  ^{2} &  =1-\left(  F_{n}^{\ast}\right)
^{4}\leq2\left(  1-\left(  F_{n}^{\ast}\right)  ^{2}\right)
\ \text{w.p.\ \ \ }1/2.
\end{align}
Defining $X_{n}$ by $X_{n}\equiv1-\left(  F_{n}^{\ast}\right)  ^{2}$, it is
now clear the process $X_{n}$ satisfies the conditions in
(\ref{eq:general-process-relations-4-conv}%
-\ref{eq:general-process-relations-4-conv-2}). Since we know that $F_{n}%
^{\ast}$ converges almost surely to a random variable $F_{\infty}^{\ast}$
taking values in $\left\{  0,1\right\}  $ with $\Pr\left\{  F_{\infty}^{\ast
}=1\right\}  =1-I\left(  W^{\ast}\right)  $, it follows that $X_{n}$ converges
almost surely to $X_{\infty}$ with $\Pr\left\{  X_{\infty}=0\right\}
=1-I\left(  W^{\ast}\right)  $. The process $X_{n}$ then satisfies all the
requirements needed to apply Theorem~\ref{thm:conv-rate}, so that%
\[
\lim_{n\rightarrow\infty}\Pr\left\{  X_{n}<2^{-2^{n\beta}}\right\}
=1-I\left(  W^{\ast}\right)  ,
\]
which in turn, from the relation $X_{n}=1-\left(  F_{n}^{\ast}\right)
^{2}\geq1-F_{n}^{\ast}$, implies that%
\[
\lim_{n\rightarrow\infty}\Pr\left\{  1-F_{n}^{\ast}<2^{-2^{n\beta}}\right\}
=1-I\left(  W^{\ast}\right)  ,
\]
giving the statement of the proposition.
\end{proof}

It is worthwhile to discuss why we specialized the above proposition to the
case where the states $\rho_{0}^{E}$ and $\rho_{1}^{E}$ are commuting. First,
as we demonstrate in Appendix~\ref{app:channel-examples}, there are many
examples of natural quantum channels for which this condition holds, including
amplitude damping channels, photon-detected jump channels, dephasing channels,
erasure channels, and cloning channels. Thus, the quantum wiretap polar coding
scheme in this section and the quantum polar coding scheme in the next section
works well for these channels. On the other hand, there exist quantum wiretap
channels for which the critical inequality in (\ref{eq:critical-inequality})
does not hold. For example, Appendix~\ref{app:pure-state-combining}%
\ demonstrates a violation of the inequality whenever the states $\rho_{0}%
^{E}$ and $\rho_{1}^{E}$ are pure and such that Tr$\left\{  \rho_{0}^{E}%
\rho_{1}^{E}\right\}  \notin\left\{  0,1\right\}  $. So the scheme given in
this section does not necessarily achieve the symmetric private capacity for
such channels because it is not clear how to guarantee that the fraction of
bad channels for Eve is equal to $1-I\left(  W^{\ast}\right)  $.

We can now establish our scheme for a quantum wiretap polar code. We divide
the set $\left[  N\right]  $ into four different subsets:%
\begin{align*}
\mathcal{A}  &  \equiv\mathcal{P}\left(  W^{\ast},\beta\right)  \cap
\mathcal{G}_{N}\left(  W,\beta\right)  ,\\
\mathcal{B}  &  \equiv\mathcal{P}\left(  W^{\ast},\beta\right)  \cap
\mathcal{B}_{N}\left(  W,\beta\right)  ,\\
\mathcal{X}  &  \equiv\mathcal{P}^{c}\left(  W^{\ast},\beta\right)
\cap\mathcal{B}_{N}\left(  W,\beta\right)  ,\\
\mathcal{Y}  &  \equiv\mathcal{P}^{c}\left(  W^{\ast},\beta\right)
\cap\mathcal{G}_{N}\left(  W,\beta\right)  .
\end{align*}
Observe that $\mathcal{A}$, $\mathcal{B}$, $\mathcal{X}$, and $\mathcal{Y}$
form a partition of $\left[  N\right]  $ because they are all pairwise
disjoint and $\mathcal{A}\cup\mathcal{B}\cup\mathcal{X}\cup\mathcal{Y}=\left[
N\right]  $. Thus, the set $\mathcal{A}$ consists of channels that are good
for Bob and bad for Eve, $\mathcal{B}$ has the channels that are bad for both,
$\mathcal{X}$ has the channels that are good for Eve and bad for Bob, and
$\mathcal{Y}$ has the channels that are good for Eve and good for Bob. The
Mahdavifar-Vardy coding scheme is then straightforward:

\begin{enumerate}
\item Send the information bits through the channels in $\mathcal{A}$.

\item Send the frozen bit vector $u_{\mathcal{B}}$ through the channels in
$\mathcal{B}$.

\item Send randomized bits through the channels in $\mathcal{Y}$.

\item We suppose that Alice and Bob have access to a secret key before
communication begins. Alice inputs her half of the secret key into the
channels in $\mathcal{X}$.\footnote{This is a slight variation of the
Mahdavifar-Vardy coding scheme that ensures both reliability and strong
security. Mahdavifar and Vardy were inconclusive about the reliability of
their coding scheme because they were unable to make any statements about the
reliability of the channels in $\mathcal{X}$. This minor variation with the
addition of a secret key ensures security and reliability because a secret key
is a hybrid of a frozen bit and a randomized bit. It is similar to a frozen
bit in that its value is available to Bob and thus he does not need to decode
the bit channels with secret key input. It is similar to a randomized bit from
the assumption that its value is uniform and unknown to Eve.} Mahdavifar and
Vardy demonstrated that the fraction $\left\vert \mathcal{X}\right\vert /N$
tends to zero in the limit $N\rightarrow\infty$ \cite{MV10}, and a slight
modification of their argument demonstrates that the codes constructed here
have the same property. This implies that the rate of secret key needed to
ensure reliability and strong security for this coding scheme vanishes and is
thus negligible in the asymptotic limit (we require the strong security
criterion for when we produce quantum polar codes from quantum wiretap polar codes).
\end{enumerate}

The following theorem guarantees that the rate of the quantum wiretap polar
code is equal to the symmetric private information:

\begin{theorem}
For the quantum wiretap polar coding scheme given above, a degraded quantum
wiretap channel with $W$ and $W^{\ast}$ as defined in (\ref{eq:channel-to-Bob}%
-\ref{eq:channel-to-Eve}), with $W^{\ast}$ having a classical output, and for
sufficiently large $N$, its rate $R_{N}=\left\vert \mathcal{A}\right\vert /N$
converges to the symmetric private information:%
\[
\lim_{N\rightarrow\infty}R_{N}=I\left(  W\right)  -I\left(  W^{\ast}\right)
.
\]

\end{theorem}

\begin{proof}
We just need to determine the size of the set $\mathcal{A}$. From basic set
theory, we know that%
\begin{align*}
\frac{\left\vert \mathcal{A}\right\vert }{N} &  =\frac{1}{N}\left\vert
\mathcal{P}\left(  W^{\ast},\beta\right)  \cap\mathcal{G}_{N}\left(
W,\beta\right)  \right\vert \\
&  =\frac{\left\vert \mathcal{P}\left(  W^{\ast},\beta\right)  \right\vert
}{N}+\frac{\left\vert \mathcal{G}_{N}\left(  W,\beta\right)  \right\vert }%
{N}\\
&  \ \ \ \ \ \ \ -\frac{1}{N}\left\vert \mathcal{P}\left(  W^{\ast}%
,\beta\right)  \cup\mathcal{G}_{N}\left(  W,\beta\right)  \right\vert .
\end{align*}
Consider that%
\begin{align*}
\mathcal{P}\left(  W^{\ast},\beta\right)  \cup\mathcal{G}_{N}\left(
W,\beta\right)   &  =\left[  N\right]  \setminus\left(  \mathcal{P}\left(
W^{\ast},\beta\right)  \cup\mathcal{G}_{N}\left(  W,\beta\right)  \right)
^{c}\\
&  =\left[  N\right]  \setminus\left(  \mathcal{P}^{c}\left(  W^{\ast}%
,\beta\right)  \cap\mathcal{B}_{N}\left(  W,\beta\right)  \right)  \\
&  =\left[  N\right]  \setminus\mathcal{X}%
\end{align*}
So it follows that%
\begin{align*}
\frac{\left\vert \mathcal{A}\right\vert }{N} &  =\frac{\left\vert
\mathcal{P}\left(  W^{\ast},\beta\right)  \right\vert }{N}+\frac{\left\vert
\mathcal{G}_{N}\left(  W,\beta\right)  \right\vert }{N}-\frac{1}{N}\left\vert
\left[  N\right]  \setminus\mathcal{X}\right\vert \\
&  =\frac{\left\vert \mathcal{P}\left(  W^{\ast},\beta\right)  \right\vert
}{N}+\frac{\left\vert \mathcal{G}_{N}\left(  W,\beta\right)  \right\vert }%
{N}-1+\frac{\left\vert \mathcal{X}\right\vert }{N}.
\end{align*}
In the limit as $N$ becomes large, we know from
Proposition~\ref{prop:poor-channel-fraction} that%
\[
\lim_{N\rightarrow\infty}\frac{\left\vert \mathcal{P}\left(  W^{\ast}%
,\beta\right)  \right\vert }{N}=1-I\left(  W^{\ast}\right)  ,
\]
and from Theorem~\ref{thm:fraction-good} that%
\[
\lim_{N\rightarrow\infty}\frac{\left\vert \mathcal{G}_{N}\left(
W,\beta\right)  \right\vert }{N}=I\left(  W\right)  .
\]
Finally, we later show that $\lim_{N\rightarrow\infty}\left\vert
\mathcal{X}\right\vert /N=0$. The statement of the theorem then follows.
\end{proof}

The following theorem demonstrates that the quantum wiretap polar coding
scheme has strong security:

\begin{theorem}
For the quantum wiretap polar coding scheme given above, a degraded quantum
wiretap channel with $W$ and $W^{\ast}$ as defined in (\ref{eq:channel-to-Bob}%
-\ref{eq:channel-to-Eve}), with $W^{\ast}$ having a classical output, and for
sufficiently large $N$, it satisfies the following strong security criterion:%
\[
I\left(  U_{\mathcal{A}};E^{n}\right)  =o\left(  2^{-\frac{1}{2}N^{\beta}%
}\right)  .
\]

\end{theorem}

\begin{proof}
Consider that%
\begin{align*}
I\left(  U_{\mathcal{A}};E^{n}\right)   &  =\sum_{i\in\mathcal{A}}I\left(
U_{i};E^{n}|U_{\mathcal{A}_{i}^{-}}\right) \\
&  =\sum_{i\in\mathcal{A}}I\left(  U_{i};E^{n}U_{\mathcal{A}_{i}^{-}}\right)
\\
&  \leq\sum_{i\in\mathcal{A}}I\left(  U_{i};E^{n}U_{1}^{i-1}\right) \\
&  =\sum_{i\in\mathcal{A}}I(W_{N}^{\ast\left(  i\right)  })
\end{align*}
The first equality is from the chain rule for quantum mutual information and
by defining $\mathcal{A}_{i}^{-}$ to be the indices in $\mathcal{A}$ preceding
$i$. The second equality follows from the assumption that the bits in
$U_{\mathcal{A}_{i}^{-}}$ are chosen uniformly at random. The firts inequality
is from quantum data processing. The third equality is from the definition of
the synthesized channels $W_{N}^{\ast\left(  i\right)  }$. Continuing, we have%
\begin{align*}
&  \leq\sum_{i\in\mathcal{A}}\sqrt{1-F(W_{N}^{\ast\left(  i\right)  })}\\
&  \leq\sum_{i\in\mathcal{A}}\sqrt{1-\left(  1-2^{-N^{\beta}}\right)  ^{2}}\\
&  =o\left(  2^{-\frac{1}{2}N^{\beta}}\right)  .
\end{align*}
The first inequality is from Proposition~1 in Ref.~\cite{WG11}. The final
inequality follows from the definition of the set $\mathcal{A}$.
\end{proof}

We also know that the code has good reliability, in the sense that there
exists a POVM\ $\left\{  \Lambda_{u_{\mathcal{A}},u_{\mathcal{Y}}}^{\left(
u_{\mathcal{X}}\right)  }\right\}  $ such that%
\begin{align*}
\Pr\{\widehat{U}_{\mathcal{A}\cup\mathcal{Y}}\neq U_{\mathcal{A}%
\cup\mathcal{Y}}\}  &  \leq2\sqrt{\sum_{i\in\mathcal{A}\cup\mathcal{Y}}%
\frac{1}{2}\sqrt{F(W_{N}^{\left(  i\right)  })}}\\
&  =o\left(  2^{-\frac{1}{2}N^{\beta}}\right)  .
\end{align*}
This POVM\ is the quantum successive cancellation decoder established in
Ref.~\cite{WG11}. The quantum successive cancellation decoder operates exactly
as before, but it needs to decode both the information bits in $\mathcal{A}$
and the randomized bits in $\mathcal{Y}$. It also exploits the frozen bits in
$\mathcal{B}$ and the secret key bits in $\mathcal{X}$ to help with decoding.

Finally, we can prove that the rate of secret key bits required by the scheme
vanishes in the limit as $N$ becomes large:

\begin{proposition}
For the quantum wiretap polar coding scheme given above, a degraded quantum
wiretap channel with $W$ and $W^{\ast}$ as defined in (\ref{eq:channel-to-Bob}%
-\ref{eq:channel-to-Eve}), with $W^{\ast}$ having a classical output, the rate
$\left\vert \mathcal{X}\right\vert /N$ vanishes as $N$ becomes large:%
\[
\lim_{N\rightarrow\infty}\frac{\left\vert \mathcal{X}\right\vert }{N}=0.
\]

\end{proposition}

\begin{proof}
This result follows by an argument similar to that for Proposition~22 in
Ref.~\cite{MV10}, but we need to modify it slightly. We prove that the sets
$\mathcal{X}$, $\mathcal{G}_{N}\left(  W^{\ast},\beta\right)  $, and
$\mathcal{P}_{N}\left(  W^{\ast},\beta\right)  $ are pairwise disjoint (note
that we define the set $\mathcal{G}_{N}\left(  W^{\ast},\beta\right)  $ as in
(\ref{eq:good-channels}), but with respect to the channel $W^{\ast}$). It then
follows that%
\begin{equation}
\frac{\left\vert \mathcal{X}\right\vert }{N}+\frac{\left\vert \mathcal{G}%
_{N}\left(  W^{\ast},\beta\right)  \right\vert }{N}+\frac{\left\vert
\mathcal{P}_{N}\left(  W^{\ast},\beta\right)  \right\vert }{N}\leq1.
\label{eq:bound-on-X}%
\end{equation}
So we prove that these sets are disjoint. First, consider that $\mathcal{X}$
and $\mathcal{P}_{N}\left(  W^{\ast},\beta\right)  $ are disjoint by
definition because $\mathcal{X}$ is formed from an intersection with
$\mathcal{P}_{N}^{c}\left(  W^{\ast},\beta\right)  $. Next, observe that for
sufficiently large $N$, $\mathcal{P}_{N}\left(  W^{\ast},\beta\right)  $ and
$\mathcal{G}_{N}\left(  W^{\ast},\beta\right)  $ are disjoint by definition.
Observe that $\mathcal{X}$ and $\mathcal{G}_{N}\left(  W^{\ast},\beta\right)
$ are disjoint because%
\[
\mathcal{B}_{N}\left(  W,\beta\right)  \subseteq\mathcal{B}_{N}\left(
W^{\ast},\beta\right)  ,
\]
which follows from $W^{\ast}$ being a degraded version of $W^{\ast}$ and
Lemma~\ref{lem:degrade} (also, $\mathcal{B}_{N}\left(  W^{\ast},\beta\right)
$ is the complement of $\mathcal{G}_{N}\left(  W^{\ast},\beta\right)  $). By
applying Proposition~\ref{prop:poor-channel-fraction}, we know that%
\[
\lim_{N\rightarrow\infty}\frac{\left\vert \mathcal{P}\left(  W^{\ast}%
,\beta\right)  \right\vert }{N}=1-I\left(  W^{\ast}\right)  ,
\]
and from Theorem~\ref{thm:fraction-good}, we know that%
\[
\lim_{N\rightarrow\infty}\frac{\left\vert \mathcal{G}_{N}\left(  W^{\ast
},\beta\right)  \right\vert }{N}=I\left(  W^{\ast}\right)  .
\]
Thus, the statement of the proposition follows from (\ref{eq:bound-on-X}) and
the above asymptotic limits.
\end{proof}

\section{Quantum polar codes}

\label{sec:quantum-polar}From such a scheme for private classical
communication over a quantum wiretap channel with classical environment, we
can readily construct a quantum polar code achieving the coherent information
of a degradable quantum channel by exploiting Devetak's ideas for quantum
coding \cite{D03} and the recent quantum successive cancellation decoder from
Ref.~\cite{WG11}. First, recall that a quantum channel $W$ is specified by a
completely-positive trace-preserving map (we consider quantum channels with
qubit inputs in this work). Any such map has a dilation to a larger system in
which the dynamics over a tensor product space are unitary, i.e., it holds
that%
\[
W\left(  \rho\right)  =\text{Tr}_{E}\left\{  U_{W}^{A^{\prime}\rightarrow
BE}\rho(U_{W}^{A^{\prime}\rightarrow BE})^{\dag}\right\}  ,
\]
where $U_{W}^{A^{\prime}\rightarrow BE}$ is the isometric extension of the
channel $W$. The complementary channel $W^{\ast}$ is the map obtained by
tracing over Bob's system%
\[
W^{\ast}\left(  \rho\right)  =\text{Tr}_{B}\left\{  U_{W^{\ast}}^{A^{\prime
}\rightarrow BE}\rho(U_{W^{\ast}}^{A^{\prime}\rightarrow BE})^{\dag}\right\}
.
\]
Such a realization makes the quantum coding setting analogous to the quantum
wiretap setting. We also define the symmetric coherent information of the
channel as%
\[
I_{c}\left(  W\right)  =H\left(  B\right)  -H\left(  AB\right)  ,
\]
where the entropies result from sending half of a maximally entangled Bell
state through the input of the channel. It is straightforward to verify that%
\[
I_{c}\left(  W\right)  =I\left(  W\right)  -I\left(  W^{\ast}\right)  .
\]
Our strategy for achieving the coherent information is to operate the quantum
wiretap polar code in superposition (just as Devetak does~\cite{D03}).

We can similarly specify a quantum polar code by the parameter vector $\left(
N,K,\mathcal{A},\mathcal{B},\mathcal{X},\mathcal{Y},u_{\mathcal{B}}\right)  $
where these parameters are all the same as in the quantum wiretap polar code.
The encoder is a coherent version of Arikan's encoder where we replace
classical CNOT\ gates with quantum CNOT gates that act as follows on a
two-qubit state:%
\begin{equation}
\sum_{x,y}\alpha_{x,y}\left\vert x\right\rangle \left\vert y\right\rangle
\rightarrow\sum_{x,y}\alpha_{x,y}\left\vert x\right\rangle \left\vert x\oplus
y\right\rangle . \label{eq:CNOT-gate}%
\end{equation}
It is important to choose the orthonormal basis for the CNOT to be the one
such that the induced states for the environment commute. That is, consider
the following induced classical-quantum channel for the environment Eve:%
\[
x\rightarrow W^{\ast}\left(  \left\vert x\right\rangle \left\langle
x\right\vert \right)  \equiv\rho_{x}^{E},
\]
where the basis $\left\{  \left\vert x\right\rangle \right\}  $ is the same as
in (\ref{eq:CNOT-gate}). The scheme in this section works at the claimed rates
if $\rho_{0}^{E}$ and $\rho_{1}^{E}$ commute (so that we can exploit the
result in Proposition~\ref{prop:poor-channel-fraction}). We prove in
Appendix~\ref{app:channel-examples} that many important channels satisfy this criterion.

We can now state our quantum polar coding theorem:

\begin{theorem}
[Quantum Polar Coding]For any degradable qubit-input quantum channel $W$ with
classical environment, there exists a quantum polar coding scheme for
entanglement generation that achieves the symmetric coherent information, in
the sense that the fidelity between the input entanglement and the generated
entanglement is equal to $1-o\left(  2^{-\frac{1}{4}N^{\beta}}\right)  $ where
$N$ is the blocklength of the code and $\beta$ is some real such that
$0<\beta<1/2$. The scheme may require entanglement assistance, but the
entanglement consumption rate vanishes in the limit of large blocklength.
\end{theorem}

\begin{proof}
We assume that our task is merely to generate entanglement between Alice and
Bob.\footnote{The task of entanglement generation is equivalent to the task of
quantum communication if forward classical communication from sender to
receiver is available. Furthermore, forward classical communication does not
increase the capacity of a quantum channel~\cite{BKN98,W11}.} Alice begins by
preparing Bell states locally on her side of the channel. Also, we assume that
Alice and Bob share a small number of ebits before communication begins. We
have the following structure for our code:

\begin{enumerate}
\item Alice sends half of the locally prepared Bell states through the
channels in $\mathcal{A}$.

\item Alice sends the frozen ancilla qubits $\left\vert u_{\mathcal{B}%
}\right\rangle $ through the channels in $\mathcal{B}$.

\item Alice sends $\left\vert \pm \right\rangle =\left(  \left\vert
0\right\rangle \pm \left\vert 1\right\rangle \right)  /\sqrt{2}$ states through
the channels in $\mathcal{Y}$ (these are bits frozen in the Hadamard basis).

\item Alice sends her shares of the ebits through the channels in
$\mathcal{X}$.
\end{enumerate}

Thus, the state before it is input to the encoder is as follows:%
\[
\frac{1}{\sqrt{2^{\left\vert A\right\vert }}}\sum_{u_{\mathcal{A}}}\left\vert
u_{\mathcal{A}}\right\rangle \left\vert u_{\mathcal{A}}\right\rangle
\left\vert u_{\mathcal{B}}\right\rangle \frac{1}{\sqrt{2^{\left\vert
\mathcal{Y}\right\vert }}}\sum_{u_{\mathcal{Y}}}\left\vert u_{\mathcal{Y}%
}\right\rangle \frac{1}{\sqrt{2^{\left\vert \mathcal{X}\right\vert }}}%
\sum_{u_{\mathcal{X}}}\left\vert u_{\mathcal{X}}\right\rangle \left\vert
u_{\mathcal{X}}\right\rangle ,
\]
where Alice possesses both shares of $\frac{1}{\sqrt{2^{\left\vert
A\right\vert }}}\sum_{u_{\mathcal{A}}}\left\vert u_{\mathcal{A}}\right\rangle
\left\vert u_{\mathcal{A}}\right\rangle $, Alice possesses $\left\vert
u_{\mathcal{B}}\right\rangle $ and the superposed state $\frac{1}%
{\sqrt{2^{\left\vert \mathcal{Y}\right\vert }}}\sum_{u_{\mathcal{Y}}%
}\left\vert u_{\mathcal{Y}}\right\rangle $, and Alice and Bob share the
entangled state $\frac{1}{\sqrt{2^{\left\vert \mathcal{X}\right\vert }}}%
\sum_{u_{\mathcal{X}}}\left\vert u_{\mathcal{X}}\right\rangle \left\vert
u_{\mathcal{X}}\right\rangle $. We can also write the above state as%
\[
\frac{1}{\sqrt{2^{\left\vert \mathcal{A}\right\vert +\left\vert \mathcal{Y}%
\right\vert +\left\vert \mathcal{X}\right\vert }}}\sum_{u_{\mathcal{A}%
},u_{\mathcal{Y}},u_{\mathcal{X}}}\left\vert u_{\mathcal{A}}\right\rangle
\left\vert u_{\mathcal{A}}\right\rangle \left\vert u_{\mathcal{B}%
}\right\rangle \left\vert u_{\mathcal{Y}}\right\rangle \left\vert
u_{\mathcal{X}}\right\rangle \left\vert u_{\mathcal{X}}\right\rangle ,
\]
and we furthermore require Alice to apply some gates that realize some
relative phases $\gamma_{u_{\mathcal{Y}}}$ so that the above state becomes%
\[
\frac{1}{\sqrt{2^{\left\vert \mathcal{A}\right\vert +\left\vert \mathcal{Y}%
\right\vert +\left\vert \mathcal{X}\right\vert }}}\sum_{u_{\mathcal{A}%
},u_{\mathcal{Y}},u_{\mathcal{X}}}e^{i\gamma_{u_{\mathcal{Y}}}}\left\vert
u_{\mathcal{A}}\right\rangle \left\vert u_{\mathcal{A}}\right\rangle
\left\vert u_{\mathcal{B}}\right\rangle \left\vert u_{\mathcal{Y}%
}\right\rangle \left\vert u_{\mathcal{X}}\right\rangle \left\vert
u_{\mathcal{X}}\right\rangle .
\]
(It is possible to realize these phases with only linear overhead in the
encoding. Also, we specify how to choose these phases later.) Alice then
applies a coherent version of Arikan's CNOT\ encoder \cite{A09}, leading to
the following encoded state:%
\[
\frac{1}{\sqrt{2^{\left\vert \mathcal{A}\right\vert }}}\sum_{u_{\mathcal{A}}%
}\left\vert u_{\mathcal{A}}\right\rangle \left\vert \phi_{u_{\mathcal{A}}%
}\right\rangle ,
\]
where $\left\{  \left\vert \phi_{u_{\mathcal{A}}}\right\rangle \right\}
_{u_{\mathcal{A}}}$ are the \textit{entanglement-assisted quantum codewords},
given by%
\begin{align*}
\left\vert \phi_{u_{\mathcal{A}}}\right\rangle  &  =\frac{1}{\sqrt
{2^{\left\vert \mathcal{Y}\right\vert +\left\vert \mathcal{X}\right\vert }}%
}\sum_{u_{\mathcal{Y}},u_{\mathcal{X}}}e^{i\gamma_{u_{\mathcal{Y}}}}\left\vert
\psi_{u_{\mathcal{A}},u_{\mathcal{B}},u_{\mathcal{Y}},u_{\mathcal{X}}%
}\right\rangle \left\vert u_{\mathcal{X}}\right\rangle \\
&  \equiv U\frac{1}{\sqrt{2^{\left\vert \mathcal{Y}\right\vert +\left\vert
\mathcal{X}\right\vert }}}\sum_{u_{\mathcal{Y}},u_{\mathcal{X}}}%
e^{i\gamma_{u_{\mathcal{Y}}}}\left\vert u_{\mathcal{A}}\right\rangle
\left\vert u_{\mathcal{B}}\right\rangle \left\vert u_{\mathcal{Y}%
}\right\rangle \left\vert u_{\mathcal{X}}\right\rangle \left\vert
u_{\mathcal{X}}\right\rangle ,
\end{align*}
with the coherent Arikan CNOT\ encoder $U$ acting only on Alice's registers.
Observe that the above state is encoded in a Calderbank-Shor-Steane code
\cite{CS96,S96}\ (modulo the relative phases) because the inputs are either
information qubits, ancillas in a fixed state $\left\vert 0\right\rangle $ or
$\left\vert 1\right\rangle $, ancillas in a state $\left\vert +\right\rangle
$, or shares of ebits, and furthermore, the encoder consists of just
SWAP\ gates and CNOT\ gates. Alice transmits the register containing the
states $\left\vert \phi_{u_{\mathcal{A}}}\right\rangle $ over the quantum
channel, leading to the following state%
\[
\frac{1}{\sqrt{2^{\left\vert \mathcal{A}\right\vert }}}\sum_{u_{\mathcal{A}}%
}\left\vert u_{\mathcal{A}}\right\rangle \left\vert \phi_{u_{\mathcal{A}}%
}\right\rangle ^{B^{N}E^{N}},
\]
where%
\[
\left\vert \phi_{u_{\mathcal{A}}}\right\rangle ^{B^{N}E^{N}}\equiv
U_{W}^{A^{N}\rightarrow B^{N}E^{N}}\left\vert \phi_{u_{\mathcal{A}}%
}\right\rangle ^{A^{N}}.
\]
Bob then applies the following coherent quantum successive cancellation
decoder to his systems $B^{N}$ and his half of the entanglement
$B^{\mathcal{X}}$:%
\begin{equation}
\sum_{u_{\mathcal{A}},u_{\mathcal{Y}},u_{\mathcal{X}}}\sqrt{\Lambda
_{u_{\mathcal{A}},u_{\mathcal{Y}}}^{\left(  u_{\mathcal{X}}\right)  }}^{B^{N}%
}\otimes\left\vert u_{\mathcal{X}}\right\rangle \left\langle u_{\mathcal{X}%
}\right\vert \otimes\left\vert u_{\mathcal{A}},u_{\mathcal{Y}}\right\rangle
^{\widehat{B}}.\label{eq:coherent-decoder}%
\end{equation}
The POVM\ elements $\left\{  \Lambda_{u_{\mathcal{A}},u_{\mathcal{Y}}%
}^{\left(  u_{\mathcal{X}}\right)  }\right\}  _{u_{\mathcal{A}},u_{\mathcal{Y}%
},u_{\mathcal{X}}}$ are the same as those in the incoherent quantum successive
cancellation decoder from Ref.~\cite{WG11}. Placing them in the above
operation allows this first part of the decoder to be an isometry and for the
code to thus operate in superposition. We claim that the following state has
fidelity $1-o\left(  2^{-\frac{1}{2}N^{\beta}}\right)  $ with the state after
Bob applies the above coherent quantum successive cancellation decoder:%
\begin{multline}
\frac{1}{\sqrt{2^{\left\vert \mathcal{A}\right\vert +\left\vert \mathcal{Y}%
\right\vert +\left\vert \mathcal{X}\right\vert }}}\sum_{u_{\mathcal{A}%
},u_{\mathcal{Y}},u_{\mathcal{X}}}e^{i\delta_{u_{\mathcal{Y}}}}\left\vert
u_{\mathcal{A}}\right\rangle \left\vert \psi_{u_{\mathcal{A}},u_{\mathcal{B}%
},u_{\mathcal{Y}},u_{\mathcal{X}}}\right\rangle ^{B^{N}E^{N}}\otimes
\label{eq:detected-state}\\
\left\vert u_{\mathcal{X}}\right\rangle \left\vert u_{\mathcal{A}%
},u_{\mathcal{Y}}\right\rangle ^{\widehat{B}},
\end{multline}
where we specify the phases $e^{i\delta_{u_{\mathcal{Y}}}}$ later. Indeed,
consider the following states%
\begin{align*}
\left\vert \chi_{u_{\mathcal{Y}}}\right\rangle  &  \equiv\frac{1}%
{\sqrt{2^{\left\vert \mathcal{A}\right\vert +\left\vert \mathcal{X}\right\vert
}}}\sum_{u_{\mathcal{A}},u_{\mathcal{X}},u_{\mathcal{A}}^{\prime
},u_{\mathcal{Y}}^{\prime}}\left\vert u_{\mathcal{A}}\right\rangle \otimes\\
&  \ \ \ \ \sqrt{\Lambda_{u_{\mathcal{A}}^{\prime},u_{\mathcal{Y}}^{\prime}%
}^{\left(  u_{\mathcal{X}}\right)  }}^{B^{N}}\left\vert \psi_{u_{\mathcal{A}%
},u_{\mathcal{B}},u_{\mathcal{Y}},u_{\mathcal{X}}}\right\rangle ^{B^{N}E^{N}%
}\left\vert u_{\mathcal{X}}\right\rangle \left\vert u_{\mathcal{A}}^{\prime
},u_{\mathcal{Y}}^{\prime}\right\rangle ^{\widehat{B}},\\
\left\vert \varphi_{u_{\mathcal{Y}}}\right\rangle  &  \equiv\frac{1}%
{\sqrt{2^{\left\vert \mathcal{A}\right\vert +\left\vert \mathcal{X}\right\vert
}}}\sum_{u_{\mathcal{A}},u_{\mathcal{X}}}\left\vert u_{\mathcal{A}%
}\right\rangle \left\vert \psi_{u_{\mathcal{A}},u_{\mathcal{B}},u_{\mathcal{Y}%
},u_{\mathcal{X}}}\right\rangle ^{B^{N}E^{N}}\otimes\\
&  \ \ \ \ \left\vert u_{\mathcal{X}}\right\rangle \left\vert u_{\mathcal{A}%
},u_{\mathcal{Y}}\right\rangle ^{\widehat{B}}.
\end{align*}
Their overlap is as follows:%
\begin{multline*}
\left\langle \varphi_{u_{\mathcal{Y}}}|\chi_{u_{\mathcal{Y}}}\right\rangle
=\frac{1}{2^{\left\vert \mathcal{A}\right\vert +\left\vert \mathcal{X}%
\right\vert }}\sum_{u_{\mathcal{A}}^{\prime\prime},u_{\mathcal{X}}%
^{\prime\prime},u_{\mathcal{A}},u_{\mathcal{X}},u_{\mathcal{A}}^{\prime
},u_{\mathcal{Y}}^{\prime}}\\
\left(  \left\langle u_{\mathcal{A}}^{\prime\prime}\right\vert \left\langle
\psi_{u_{\mathcal{A}}^{\prime\prime},u_{\mathcal{B}},u_{\mathcal{Y}%
},u_{\mathcal{X}}^{\prime\prime}}\right\vert ^{B^{N}E^{N}}\left\langle
u_{\mathcal{X}}^{\prime\prime}\right\vert \left\langle u_{\mathcal{A}}%
^{\prime\prime},u_{\mathcal{Y}}\right\vert ^{\widehat{B}}\right)  \times\\
\left(  \sqrt{\Lambda_{u_{\mathcal{A}}^{\prime},u_{\mathcal{Y}}^{\prime}%
}^{\left(  u_{\mathcal{X}}\right)  }}^{B^{N}}\left\vert u_{\mathcal{A}%
}\right\rangle \left\vert \psi_{u_{\mathcal{A}},u_{\mathcal{B}},u_{\mathcal{Y}%
},u_{\mathcal{X}}}\right\rangle ^{B^{N}E^{N}}\left\vert u_{\mathcal{X}%
}\right\rangle \left\vert u_{\mathcal{A}}^{\prime},u_{\mathcal{Y}}^{\prime
}\right\rangle ^{\widehat{B}}\right)
\end{multline*}%
\begin{multline*}
=\frac{1}{2^{\left\vert \mathcal{A}\right\vert +\left\vert \mathcal{X}%
\right\vert }}\sum_{u_{\mathcal{A}}^{\prime\prime},u_{\mathcal{X}}%
^{\prime\prime},u_{\mathcal{A}},u_{\mathcal{X}},u_{\mathcal{A}}^{\prime
},u_{\mathcal{Y}}^{\prime}}\left\langle u_{\mathcal{A}}^{\prime\prime
}|u_{\mathcal{A}}\right\rangle \times\\
\left\langle \psi_{u_{\mathcal{A}}^{\prime\prime},u_{\mathcal{B}%
},u_{\mathcal{Y}},u_{\mathcal{X}}^{\prime\prime}}\right\vert \sqrt
{\Lambda_{u_{\mathcal{A}}^{\prime},u_{\mathcal{Y}}^{\prime}}^{\left(
u_{\mathcal{X}}\right)  }}^{B^{N}}\left\vert \psi_{u_{\mathcal{A}%
},u_{\mathcal{B}},u_{\mathcal{Y}},u_{\mathcal{X}}}\right\rangle ^{B^{N}E^{N}%
}\times\\
\left\langle u_{\mathcal{X}}^{\prime\prime}|u_{\mathcal{X}}\right\rangle
\left\langle u_{\mathcal{A}}^{\prime\prime}|u_{\mathcal{A}}^{\prime
}\right\rangle \left\langle u_{\mathcal{Y}}|u_{\mathcal{Y}}^{\prime
}\right\rangle
\end{multline*}%
\begin{multline*}
=\frac{1}{2^{\left\vert \mathcal{A}\right\vert +\left\vert \mathcal{X}%
\right\vert }}\times\\
\sum_{u_{\mathcal{A}},u_{\mathcal{X}}}\left\langle \psi_{u_{\mathcal{A}%
},u_{\mathcal{B}},u_{\mathcal{Y}},u_{\mathcal{X}}}\right\vert \sqrt
{\Lambda_{u_{\mathcal{A}},u_{\mathcal{Y}}}^{\left(  u_{\mathcal{X}}\right)  }%
}^{B^{N}}\left\vert \psi_{u_{\mathcal{A}},u_{\mathcal{B}},u_{\mathcal{Y}%
},u_{\mathcal{X}}}\right\rangle ^{B^{N}E^{N}}%
\end{multline*}%
\[
\geq\frac{1}{2^{\left\vert \mathcal{A}\right\vert +\left\vert \mathcal{X}%
\right\vert }}\sum_{u_{\mathcal{A}},u_{\mathcal{X}}}\text{Tr}\left\{
\Lambda_{u_{\mathcal{A}},u_{\mathcal{Y}}}^{\left(  u_{\mathcal{X}}\right)
B^{N}}\psi_{u_{\mathcal{A}},u_{\mathcal{B}},u_{\mathcal{Y}},u_{\mathcal{X}}%
}^{B^{N}}\right\}  .
\]
We can then consider the average overlap:%
\begin{align*}
&  \frac{1}{2^{\left\vert \mathcal{Y}\right\vert }}\sum_{u_{\mathcal{Y}}%
}\left\langle \varphi_{u_{\mathcal{Y}}}|\chi_{u_{\mathcal{Y}}}\right\rangle \\
&  \geq\frac{1}{2^{\left\vert \mathcal{A}\right\vert +\left\vert
\mathcal{X}\right\vert +\left\vert \mathcal{Y}\right\vert }}\sum
_{u_{\mathcal{A}},u_{\mathcal{X}},u_{\mathcal{Y}}}\text{Tr}\left\{
\Lambda_{u_{\mathcal{A}},u_{\mathcal{Y}}}^{\left(  u_{\mathcal{X}}\right)
B^{N}}\psi_{u_{\mathcal{A}},u_{\mathcal{B}},u_{\mathcal{Y}},u_{\mathcal{X}}%
}^{B^{N}}\right\}  \\
&  \geq1-o\left(  2^{-\frac{1}{2}N^{\beta}}\right)  ,
\end{align*}
where the last inequality follows from the performance of the corresponding
quantum wiretap polar code. Thus, there must exist some phases $\delta
_{u_{\mathcal{Y}}}$ and $\gamma_{u_{\mathcal{Y}}}$ such that the state in
(\ref{eq:detected-state}) has a large overlap with the state resulting from
the action of the coherent quantum successive cancellation decoder in
(\ref{eq:coherent-decoder}) (this follows from Lemma~4 of Ref.~\cite{D03}).
The state resulting from the coherent quantum successive cancellation decoder
is then close to the state in (\ref{eq:detected-state}), which we repeat
below:%
\begin{multline*}
\frac{1}{\sqrt{2^{\left\vert \mathcal{A}\right\vert +\left\vert \mathcal{Y}%
\right\vert +\left\vert \mathcal{X}\right\vert }}}\sum_{u_{\mathcal{A}%
},u_{\mathcal{Y}},u_{\mathcal{X}}}e^{i\delta_{u_{\mathcal{Y}}}}\left\vert
u_{\mathcal{A}}\right\rangle \left\vert \psi_{u_{\mathcal{A}},u_{\mathcal{B}%
},u_{\mathcal{Y}},u_{\mathcal{X}}}\right\rangle ^{B^{N}E^{N}}\otimes\\
\left\vert u_{\mathcal{X}}\right\rangle \left\vert u_{\mathcal{A}%
},u_{\mathcal{Y}}\right\rangle ^{\widehat{B}}.
\end{multline*}
For a particular value of $u_{\mathcal{A}}$, the above state is%
\begin{multline*}
\frac{1}{\sqrt{2^{\left\vert \mathcal{Y}\right\vert +\left\vert \mathcal{X}%
\right\vert }}}\sum_{u_{\mathcal{Y}},u_{\mathcal{X}}}e^{i\delta
_{u_{\mathcal{Y}}}}\left\vert u_{\mathcal{A}}\right\rangle \left\vert
\psi_{u_{\mathcal{A}},u_{\mathcal{B}},u_{\mathcal{Y}},u_{\mathcal{X}}%
}\right\rangle ^{B^{N}E^{N}}\otimes\\
\left\vert u_{\mathcal{X}}\right\rangle \left\vert u_{\mathcal{A}%
},u_{\mathcal{Y}}\right\rangle ^{\widehat{B}}.
\end{multline*}
Tracing over Bob's systems gives the following state%
\[
\psi_{u_{\mathcal{A}}}^{E^{N}}\equiv\frac{1}{2^{\left\vert \mathcal{X}%
\right\vert +\left\vert \mathcal{Y}\right\vert }}\sum_{u_{\mathcal{X}%
},u_{\mathcal{Y}}}\psi_{u_{\mathcal{A}},u_{\mathcal{B}},u_{\mathcal{Y}%
},u_{\mathcal{X}}}^{E^{N}}.
\]
We know from the quantum wiretap polar code that the following property holds%
\[
I\left(  U_{\mathcal{A}};E^{n}\right)  =o\left(  2^{-\frac{1}{2}N^{\beta}%
}\right)  .
\]
Thus, from the fact that relative entropy upper bounds the trace distance
(Theorem~11.9.2 of Ref.~\cite{W11}), we know that%
\begin{align*}
2\ln2\sqrt{I\left(  U_{\mathcal{A}};E^{n}\right)  } &  \geq\left\Vert
\rho^{U_{\mathcal{A}}}\otimes\rho^{E^{n}}-\rho^{U_{\mathcal{A}}E^{n}%
}\right\Vert _{1}\\
&  =\frac{1}{2^{\left\vert \mathcal{A}\right\vert }}\sum_{u_{\mathcal{A}}%
}\left\Vert \psi_{u_{\mathcal{A}}}^{E^{N}}-\psi^{E^{N}}\right\Vert _{1},
\end{align*}
where%
\[
\psi^{E^{N}}\equiv\frac{1}{2^{\left\vert \mathcal{A}\right\vert }}%
\sum_{u_{\mathcal{A}}}\psi_{u_{\mathcal{A}}}^{E^{N}}.
\]
Thus, it follows that%
\[
\frac{1}{2^{\left\vert \mathcal{A}\right\vert }}\sum_{u_{\mathcal{A}}%
}\left\Vert \psi_{u_{\mathcal{A}}}^{E^{N}}-\psi^{E^{N}}\right\Vert
_{1}=o\left(  2^{-\frac{1}{4}N^{\beta}}\right)  ,
\]
and from the relationship between trace distance and fidelity~\cite{W11}, we
have that%
\[
\frac{1}{2^{\left\vert \mathcal{A}\right\vert }}\sum_{u_{\mathcal{A}}}%
F(\psi_{u_{\mathcal{A}}}^{E^{N}},\psi^{E^{N}})=1-o\left(  2^{-\frac{1}%
{4}N^{\beta}}\right)  .
\]
So, for each $u_{\mathcal{A}}$, consider that Bob possesses the purification
of the state $\psi_{u_{\mathcal{A}}}^{E^{N}}$ or $\psi^{E^{N}}$ and by
Uhlmann's theorem~\cite{U73,W11}, there exist isometries $U_{u_{\mathcal{A}}}%
$\ such that%
\[
F(\psi_{u_{\mathcal{A}}}^{E^{N}},\psi^{E^{N}})=F(U_{u_{\mathcal{A}}}%
\psi_{u_{\mathcal{A}}}^{B^{N}E^{N}}U_{u_{\mathcal{A}}}^{\dag},\psi^{B^{N}%
E^{N}}),
\]
for all $u_{\mathcal{A}}$, where we observe that $\psi_{u_{\mathcal{A}}%
}^{B^{N}E^{N}}$ is a purification of $\psi_{u_{\mathcal{A}}}^{E^{N}}$, defined
as%
\begin{multline*}
\left\vert \psi_{u_{\mathcal{A}}}\right\rangle ^{B^{N}E^{N}}=\\
\sum_{u_{\mathcal{Y}},u_{\mathcal{X}}}\frac{1}{\sqrt{2^{\left\vert
\mathcal{Y}\right\vert +\left\vert \mathcal{X}\right\vert }}}e^{i\delta
_{u_{\mathcal{Y}}}}\left(  \left\vert \psi_{u_{\mathcal{A}},u_{\mathcal{B}%
},u_{\mathcal{Y}},u_{\mathcal{X}}}\right\rangle ^{B^{N}E^{N}}\left\vert
u_{\mathcal{X}}\right\rangle \left\vert u_{\mathcal{Y}}\right\rangle
^{\widehat{B}}\right)  .
\end{multline*}
Thus, Bob's final task is to apply the controlled unitary on his state%
\[
\sum_{u_{\mathcal{A}}}\left\vert u_{\mathcal{A}}\right\rangle \left\langle
u_{\mathcal{A}}\right\vert ^{\widehat{B}}\otimes U_{u_{\mathcal{A}}},
\]
leading to the following state%
\begin{multline*}
\frac{1}{\sqrt{2^{\left\vert \mathcal{A}\right\vert }}}\sum_{u_{\mathcal{A}}%
}\left\vert u_{\mathcal{A}}\right\rangle \left\vert u_{\mathcal{A}%
}\right\rangle ^{\widehat{B}}U_{u_{\mathcal{A}}}\sum_{u_{\mathcal{Y}%
},u_{\mathcal{X}}}\frac{1}{\sqrt{2^{\left\vert \mathcal{Y}\right\vert
+\left\vert \mathcal{X}\right\vert }}}e^{i\delta_{u_{\mathcal{Y}}}}\times\\
\left(  \left\vert \psi_{u_{\mathcal{A}},u_{\mathcal{B}},u_{\mathcal{Y}%
},u_{\mathcal{X}}}\right\rangle ^{B^{N}E^{N}}\left\vert u_{\mathcal{X}%
}\right\rangle \left\vert u_{\mathcal{Y}}\right\rangle ^{\widehat{B}}\right)
.
\end{multline*}
The desired, ideal state is as follows:%
\[
\frac{1}{\sqrt{2^{\left\vert \mathcal{A}\right\vert }}}\sum_{u_{\mathcal{A}}%
}\left\vert u_{\mathcal{A}}\right\rangle \left\vert \psi\right\rangle
^{B^{N}E^{N}}\left\vert u_{\mathcal{A}}\right\rangle ^{\widehat{B}},
\]
and, after a straightforward calculation, the overlap between the actual state
and the desired state is equal to%
\begin{align*}
& \frac{1}{2^{\left\vert \mathcal{A}\right\vert }}\sum_{u_{\mathcal{A}}%
}F(U_{u_{\mathcal{A}}}\psi_{u_{\mathcal{A}}}^{B^{N}E^{N}}U_{u_{\mathcal{A}}%
}^{\dag},\psi^{B^{N}E^{N}})\\
& =\frac{1}{2^{\left\vert \mathcal{A}\right\vert }}\sum_{u_{\mathcal{A}}%
}F(\psi_{u_{\mathcal{A}}}^{E^{N}},\psi^{E^{N}})\\
& =1-o\left(  2^{-\frac{1}{4}N^{\beta}}\right)  .
\end{align*}
Bob's final move is to discard his register in $B^{N}$, leading to a state
close to the following entangled state shared between Alice and Bob:%
\[
\frac{1}{\sqrt{2^{\left\vert \mathcal{A}\right\vert }}}\sum_{u_{\mathcal{A}}%
}\left\vert u_{\mathcal{A}}\right\rangle \left\vert u_{\mathcal{A}%
}\right\rangle ^{\widehat{B}}%
\]
Putting everything together, we have a scheme that generates entanglement with
a fidelity equal to%
\[
1-o\left(  2^{-\frac{1}{4}N^{\beta}}\right)  ,
\]
where the error results from two parts in the protocol:\ the first error is
from the reliability of the quantum wiretap polar code and the second is from
the strong security of the quantum wiretap polar code.
\end{proof}

The above scheme achieves the symmetric coherent information rate and has an
efficient encoder, but it is unclear to us right now how to improve the
efficiency of the two-step decoder. If one were able to improve the efficiency
of the quantum successive cancellation decoder, this would lead to an
efficient implementation of the first part of the decoding. Improving the
efficiency of the second part might be possible by taking into account the
particular structure of these quantum polar codes.

As a final point, we should note that the symmetric coherent information is
equal to the quantum capacity of the quantum erasure channel, the quantum dephasing channel,
and the universal cloning machine channel. This is not the case for the amplitude damping channel
and the photon-detected jump channel, but the ratio between the true quantum capacity
and the symmetric coherent information is close to one for these channels (see Figure~\ref{fig:q-cap-ratio}).
\begin{figure}[ptb]
\begin{center}
\includegraphics[
width=3.4in
]{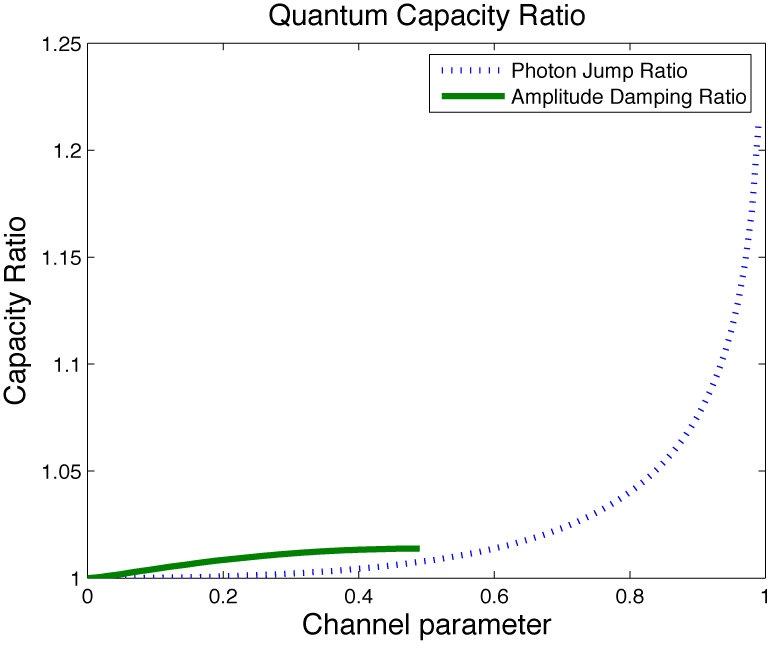}
\end{center}
\caption{The figure plots the ratio of the true quantum capacity to the symmetric
coherent information as a function of the channel parameter for both the amplitude damping channel and the photon-detected jump channel (see Appendix~\ref{app:channel-examples} for definitions of these channels). Observe that the ratio is close to one for many values of the channel parameters.}%
\label{fig:q-cap-ratio}%
\end{figure}

\section{Conclusion}

We constructed polar codes for transmitting classical data privately over a
quantum wiretap channel and for transmitting quantum data over a quantum
channel. The rates achievable are respectively equal to the symmetric private
information and the symmetric coherent information. In both settings, we
require that the channel to the environment be effectively classical in order
to guarantee that the rates are as claimed. The codes exploit the channel
polarization phenomenon, first observed by Arikan in the classical setting
\cite{A09}\ and later in Ref.~\cite{WG11}\ for classical-quantum channels.

Two of the most important problems left open in this paper
are to determine an efficient quantum decoder and to find
other channels outside of the
class discussed here for which the symmetric coherent information is achievable.
Progress on the first question is outlined in Ref.~\cite{WLH13}, though the main question
of an efficient decoder still remains open. The recent work of Wilde and Renes resolves the
second question \cite{WR12,WR12a,RW12}, by adapting an earlier coding scheme
of Renes and Boileau \cite{RB08} to
the polar coding setting.

One might question whether we should call the codes developed here and in
Ref.~\cite{RDR11}\ \textquotedblleft quantum polar codes,\textquotedblright%
\ if the criterion for a quantum polar code is that it be symmetric
capacity-achieving, channel-adapted, and possessing efficient encoders and
decoders. The codes constructed here are symmetric capacity-achieving and
channel-adapted, but do not have efficient decoders. The codes from
Ref.~\cite{RDR11} have efficient encoders and decoders, but they are
capacity-achieving and channel-adapted only for Pauli channels. Though,
both constructions certainly take advantage of the channel polarization effect,
which gives the code construction its name. Solving the
open problems posed here would certainly lead to quantum polar codes that
possess all desiderata.

\section*{Acknowledgements}

MMW acknowledges financial support from the MDEIE (Qu\'{e}bec) PSR-SIIRI
international collaboration grant. SG was supported by the DARPA Information
in a Photon (InPho) program under contract number HR0011-10-C-0159. The
authors are grateful to S.~Hamed Hassani for helpful discussions concerning
Theorem~3 in Ref.~\cite{HU10} and Omar Fawzi and Patrick Hayden for useful discussions.

\appendices

\section{Fidelity of pure state cq channels remains invariant under channel
combining}

\label{app:pure-state-combining}In the theorem below, we demonstrate that the
fidelity for the \textquotedblleft worse\textquotedblright\ channel $W^{-}$
\cite{WG11} is equal to the fidelity of the original channel $W$ if the
original channel is a classical-quantum state with pure state outputs. The
implication is that the inequality in (\ref{eq:critical-inequality}) fails for
this case.

\begin{theorem}
Suppose we have two pure states $\left\vert \psi_{0}\right\rangle $ and
$\left\vert \psi_{1}\right\rangle $ such that the classical-quantum channel
$W$ outputs $\left\vert \psi_{0}\right\rangle $ if zero is input and
$\left\vert \psi_{1}\right\rangle $ if one is input. The fidelity between
these two pure states is as follows:%
\[
F\left(  \left\vert \psi_{0}\right\rangle ,\left\vert \psi_{1}\right\rangle
\right)  =\left\vert \left\langle \psi_{0}|\psi_{1}\right\rangle \right\vert
^{2}.
\]
The states arising from channel combining in the worse direction $W^{-}$
\cite{WG11} are as follows:%
\begin{align*}
\rho_{0}^{-}  &  =\frac{1}{2}\left(  \left\vert \psi_{0}\right\rangle
\left\langle \psi_{0}\right\vert \otimes\left\vert \psi_{0}\right\rangle
\left\langle \psi_{0}\right\vert +\left\vert \psi_{1}\right\rangle
\left\langle \psi_{1}\right\vert \otimes\left\vert \psi_{1}\right\rangle
\left\langle \psi_{1}\right\vert \right)  ,\\
\rho_{1}^{-}  &  =\frac{1}{2}\left(  \left\vert \psi_{0}\right\rangle
\left\langle \psi_{0}\right\vert \otimes\left\vert \psi_{1}\right\rangle
\left\langle \psi_{1}\right\vert +\left\vert \psi_{1}\right\rangle
\left\langle \psi_{1}\right\vert \otimes\left\vert \psi_{0}\right\rangle
\left\langle \psi_{0}\right\vert \right)  .
\end{align*}
Then the following relation holds%
\[
F\left(  W\right)  =F\left(  W^{-}\right)  .
\]

\end{theorem}

\begin{proof}
Recall that the Uhlmann fidelity $F\left(  \rho_{0}^{-},\rho_{1}^{-}\right)
$\ is equal to the maximum squared overlap between purifications of $\rho
_{0}^{-}$ and $\rho_{1}^{-}$, where the maximum is over all purifications.
Since both of the above states are rank two, it suffices to consider a
two-dimensional purifying system for both. Furthermore, we can fix one
purification while varying the other one. So, one purification of $\rho
_{0}^{-}$ is as follows:%
\[
|\phi_{\rho_{0}^{-}}\rangle=\frac{1}{\sqrt{2}}\left(  \left\vert \psi
_{0}\right\rangle \left\vert \psi_{0}\right\rangle \left\vert 0\right\rangle
+\left\vert \psi_{1}\right\rangle \left\vert \psi_{1}\right\rangle \left\vert
1\right\rangle \right)  ,
\]
and a varying purification of $\rho_{1}^{-}$\ is as follows:%
\[
|\phi_{\rho_{1}^{-}}\rangle=\frac{1}{\sqrt{2}}\left(  \left\vert \psi
_{0}\right\rangle \left\vert \psi_{1}\right\rangle \left\vert \varphi
\right\rangle +\left\vert \psi_{1}\right\rangle \left\vert \psi_{0}%
\right\rangle \left\vert \varphi^{\perp}\right\rangle \right)  ,
\]
where we can set%
\begin{align*}
\left\vert \varphi\right\rangle  &  =\alpha\left\vert 0\right\rangle
+\beta\left\vert 1\right\rangle ,\\
|\varphi^{\perp}\rangle &  =\beta^{\ast}\left\vert 0\right\rangle
-\alpha^{\ast}\left\vert 1\right\rangle .
\end{align*}
Our goal is to maximize the overlap $\left\vert \langle\phi_{\rho_{0}^{-}%
}|\phi_{\rho_{1}^{-}}\rangle\right\vert ^{2}$\ over all legitimate choices of
$\alpha$ and $\beta$. The overlap $\left\vert \langle\phi_{\rho_{0}^{-}}%
|\phi_{\rho_{1}^{-}}\rangle\right\vert ^{2}$ is as follows:%
\begin{align*}
&  \left\vert \langle\phi_{\rho_{0}^{-}}|\phi_{\rho_{1}^{-}}\rangle\right\vert
^{2}\\
&  =\frac{1}{4}\left\vert
\begin{array}
[c]{c}%
\left\langle \psi_{0}|\psi_{0}\right\rangle \left\langle \psi_{0}|\psi
_{1}\right\rangle \left\langle 0|\varphi\right\rangle +\left\langle \psi
_{1}|\psi_{0}\right\rangle \left\langle \psi_{1}|\psi_{1}\right\rangle
\left\langle 1|\varphi\right\rangle +\\
\left\langle \psi_{0}|\psi_{1}\right\rangle \left\langle \psi_{0}|\psi
_{0}\right\rangle \left\langle 0|\varphi^{\perp}\right\rangle +\left\langle
\psi_{1}|\psi_{1}\right\rangle \left\langle \psi_{1}|\psi_{0}\right\rangle
\left\langle 1|\varphi^{\perp}\right\rangle
\end{array}
\right\vert ^{2}\\
&  =\frac{1}{4}\left\vert
\begin{array}
[c]{c}%
\left\langle \psi_{0}|\psi_{1}\right\rangle \alpha+\left\langle \psi_{1}%
|\psi_{0}\right\rangle \beta+\\
\left\langle \psi_{0}|\psi_{1}\right\rangle \beta^{\ast}-\left\langle \psi
_{1}|\psi_{0}\right\rangle \alpha^{\ast}%
\end{array}
\right\vert ^{2}\\
&  =\frac{1}{4}\left\vert \left\langle \psi_{0}|\psi_{1}\right\rangle \left(
\alpha+\beta^{\ast}\right)  +\left\langle \psi_{1}|\psi_{0}\right\rangle
\left(  \beta-\alpha^{\ast}\right)  \right\vert ^{2}\\
&  =\frac{1}{4}\left\vert \left\langle \psi_{0}|\psi_{1}\right\rangle \left(
\alpha+\beta^{\ast}\right)  +\left\langle \psi_{1}|\psi_{0}\right\rangle
\left(  \beta-\alpha^{\ast}\right)  \right\vert ^{2}%
\end{align*}%
\begin{align*}
&  =\frac{1}{4}\left\vert \left\langle \psi_{0}|\psi_{1}\right\rangle
\right\vert ^{2}\left\vert \alpha+\beta^{\ast}\right\vert ^{2}+\frac{1}%
{4}\left\vert \left\langle \psi_{0}|\psi_{1}\right\rangle \right\vert
^{2}\left\vert \beta-\alpha^{\ast}\right\vert ^{2}\\
&  \ \ \ +\frac{1}{4}2\operatorname{Re}\left\{  \left\langle \psi_{1}|\psi
_{0}\right\rangle ^{2}\left(  \alpha^{\ast}+\beta\right)  \left(  \beta
-\alpha^{\ast}\right)  \right\}  \\
&  =\frac{1}{4}\left\vert \left\langle \psi_{0}|\psi_{1}\right\rangle
\right\vert ^{2}\left(  \left\vert \alpha+\beta^{\ast}\right\vert
^{2}+\left\vert \beta-\alpha^{\ast}\right\vert ^{2}\right)  \\
&  \ \ \ +\frac{1}{4}2\operatorname{Re}\left\{  \left\langle \psi_{1}|\psi
_{0}\right\rangle ^{2}\left(  \beta^{2}-\left(  \alpha^{\ast}\right)
^{2}\right)  \right\}  \\
&  =\frac{1}{4}\left\vert \left\langle \psi_{0}|\psi_{1}\right\rangle
\right\vert ^{2}(\left\vert \alpha\right\vert ^{2}+\left\vert \beta\right\vert
^{2}+\\
&  \ \ \ \ 2\operatorname{Re}\left\{  \alpha\beta^{\ast}\right\}  +\left\vert
\alpha\right\vert ^{2}+\left\vert \beta\right\vert ^{2}-2\operatorname{Re}%
\left\{  \alpha\beta^{\ast}\right\}  )\\
&  \ \ \ \ +\frac{1}{4}2\operatorname{Re}\left\{  \left\langle \psi_{1}%
|\psi_{0}\right\rangle ^{2}\left(  \beta^{2}-\left(  \alpha^{\ast}\right)
^{2}\right)  \right\}  \\
&  =\frac{1}{2}\left\vert \left\langle \psi_{0}|\psi_{1}\right\rangle
\right\vert ^{2}+\frac{1}{2}\operatorname{Re}\left\{  \left\langle \psi
_{1}|\psi_{0}\right\rangle ^{2}\left(  \beta^{2}-\left(  \alpha^{\ast}\right)
^{2}\right)  \right\}
\end{align*}
We set $\left\langle \psi_{1}|\psi_{0}\right\rangle =r_{\psi}e^{i\phi}$ (with
$r_{\psi}=\left\vert \left\langle \psi_{1}|\psi_{0}\right\rangle \right\vert
$) so that we have%
\begin{align*}
&  \frac{1}{2}\operatorname{Re}\left\{  \left\langle \psi_{1}|\psi
_{0}\right\rangle ^{2}\left(  \beta^{2}-\left(  \alpha^{\ast}\right)
^{2}\right)  \right\}  \\
&  =\frac{1}{2}\operatorname{Re}\left\{  r_{\psi}^{2}e^{i2\phi}\left(
r_{1}^{2}e^{i2\theta_{1}}-r_{2}^{2}e^{i2\theta_{2}}\right)  \right\}  \\
&  =\frac{1}{2}r_{\psi}^{2}\operatorname{Re}\left\{  r_{1}^{2}e^{i2\left(
\theta_{1}+\phi\right)  }-r_{2}^{2}e^{i2\left(  \theta_{2}+\phi\right)
}\right\}  \\
&  =\frac{1}{2}\left\vert \left\langle \psi_{1}|\psi_{0}\right\rangle
\right\vert ^{2}\left(  r_{1}^{2}\cos\left(  2\left(  \theta_{1}+\phi\right)
\right)  -r_{2}^{2}\cos\left(  2\left(  \theta_{2}+\phi\right)  \right)
\right)  \\
&  \leq\frac{1}{2}\left\vert \left\langle \psi_{1}|\psi_{0}\right\rangle
\right\vert ^{2}\left(  r_{1}^{2}+r_{2}^{2}\right)  \\
&  =\frac{1}{2}\left\vert \left\langle \psi_{1}|\psi_{0}\right\rangle
\right\vert ^{2}%
\end{align*}
By choosing an appropriate $\alpha$ and $\beta$ such that $\cos\left(
2\left(  \theta_{1}+\phi\right)  \right)  =1$ and $\cos\left(  2\left(
\theta_{2}+\phi\right)  \right)  =-1$, this demonstrates that the result holds
for any pure states $\left\vert \psi_{0}\right\rangle $ and $\left\vert
\psi_{1}\right\rangle $.
\end{proof}

\section{Examples of channels with classical environment}

\label{app:channel-examples}We prove in this section that several important
degradable channels have a classical environment, in the sense that the
classical-quantum channel induced by inputting classical orthonormal states at
the input leads to commuting output states. That is, we would like to prove
that $\left[  \rho_{0}^{E},\rho_{1}^{E}\right]  =0$ for several important
channels, where%
\[
\rho_{x}^{E}=W^{\ast}\left(  \left\vert x\right\rangle \left\langle
x\right\vert \right)  ,
\]
$W^{\ast}$ is the complementary channel, and $\left\{  \left\vert
x\right\rangle \right\}  $ is some orthonormal basis.

We begin with the amplitude damping channel. The complement of an amplitude
damping channel with damping parameter $\gamma$ has the following action on a
qubit input \cite{W11}:%
\[%
\begin{bmatrix}
1-p & \eta^{\ast}\\
\eta & p
\end{bmatrix}
\rightarrow%
\begin{bmatrix}
1-\gamma p & \sqrt{\gamma}\eta^{\ast}\\
\sqrt{\gamma}\eta & \gamma p
\end{bmatrix}
,
\]
where $0\leq p,\gamma\leq1$, $\eta$ is a complex number such that the input
matrix is positive, and the matrix representations are with respect to the
computational basis $\left\{  \left\vert 0\right\rangle ,\left\vert
1\right\rangle \right\}  $. (The complement is effectively an amplitude
damping channel with damping parameter $1-\gamma$.) The result follows by
observing that%
\begin{align*}
\left\vert 0\right\rangle \left\langle 0\right\vert  &  \rightarrow\left\vert
0\right\rangle \left\langle 0\right\vert ,\\
\left\vert 1\right\rangle \left\langle 1\right\vert  &  \rightarrow\left(
1-\gamma\right)  \left\vert 0\right\rangle \left\langle 0\right\vert
+\gamma\left\vert 1\right\rangle \left\langle 1\right\vert .
\end{align*}

Consider the photon-detected jump channel from Refs.~\cite{ABCDGM01,GJWZ10}.
The authors of Ref.~\cite{GJWZ10}\ demonstrated that the complement of this
channel is as follows:%
\[%
\begin{bmatrix}
1-p & \eta^{\ast}\\
\eta & p
\end{bmatrix}
\rightarrow\left(  1-\gamma p\right)  \left\vert 0\right\rangle \left\langle
0\right\vert ^{E}+\gamma p\left\vert 1\right\rangle \left\langle 1\right\vert
^{E}.
\]
So it is again clear that the computational basis suffices to make the
environment outputs commute.

The complement of an erasure channel with erasure parameter $\epsilon$ is just
an erasure channel with erasure parameter $1-\epsilon$ \cite{W11}:%
\[
\rho\rightarrow\epsilon\rho+\left(  1-\epsilon\right)  \left\vert
e\right\rangle \left\langle e\right\vert ,
\]
where $\left\vert e\right\rangle $ is some erasure symbol orthogonal to the
space of $\rho$. Thus, any basis suffices to demonstrate that the states for
the environment commute.

A qubit dephasing channel with parameter $p$ has the following form:%
\[
\rho\rightarrow\left(  1-p\right)  \rho+p\ \sigma_{i}\rho\sigma_{i},
\]
where $\sigma_{i}$ is one of the Pauli operators. The complement of this
channel has the following form for a pure state input $\left\vert
\psi\right\rangle $ \cite{W11}:%
\begin{multline*}
\left(  1-p\right)  \left\vert 0\right\rangle \left\langle 0\right\vert +\\
\sqrt{p\left(  1-p\right)  }\left\langle \psi\right\vert \sigma_{i}\left\vert
\psi\right\rangle \left(  \left\vert 0\right\rangle \left\langle 1\right\vert
+\left\vert 1\right\rangle \left\langle 0\right\vert \right)  +\\
p\left\vert 1\right\rangle \left\langle 1\right\vert .
\end{multline*}
Thus, choosing the input basis to be the one for which $\sigma_{i}$ acts as a
\textquotedblleft bit-flipping\textquotedblright\ operator leads to commuting
outputs for the environment. For example, if $\sigma_{i}=X$, then the basis is
$\left\{  \left\vert 0\right\rangle ,\left\vert 1\right\rangle \right\}  $,
while if $\sigma_{i}=Z$, the basis is $\left\{  \left\vert +\right\rangle
,\left\vert -\right\rangle \right\}  $.

Finally, the complement of the cloning channel is available in
Ref.~\cite{GJWZ10}. Due to the covariance of the cloning channel and its
complement, inputting any orthonormal basis leads to the following states on
the output:%
\begin{align}
\psi_{0}^{E}  &  =\sum_{i=0}^{N-1}\frac{i+1}{\Delta_{N}}|i\rangle\!\langle
i|^{E},\\
\psi_{1}^{E}  &  =\sum_{i=0}^{N-1}\frac{i+1}{\Delta_{N}}|N-1-i\rangle\!\langle
N-1-i|^{E},
\end{align}
where $N$ is the number of clones, $\Delta_{N}=N\left(  N+1\right)  /2$, and
the states $\{\left\vert i\right\rangle ^{E}\}$ form an orthonormal basis.
Thus, the environmental states commute for these channels.

\bibliographystyle{plain}
\bibliography{Ref}

\begin{biography}{Mark M. Wilde}(M'99) was born in Metairie, Louisiana, USA. He received the Ph.D. degree in electrical engineering from the University of Southern California, Los Angeles, California, in 2008. He is a Postdoctoral Fellow at the School of Computer Science, McGill University and will start in August 2013 as an Assistant Professor in the Department of Physics and Astronomy and the Center for Computation and Technology at Louisiana State University. His current research interests are in quantum Shannon theory, quantum optical communication, quantum computational complexity theory, and quantum error correction.\end{biography}

\begin{biography}{Saikat Guha} was born in Patna, India, on July 3, 1980. He received the Ph.D. degree in Electrical Engineering and Computer Science from the Massachusetts Institute of Technology (MIT), Cambridge, MA in 2008. He is currently a Senior Scientist with Raytheon BBN Technologies, Cambridge, MA, USA. His current research interest surrounds the application of quantum information and estimation theory to fundamental limits of optical communication and imaging. He is also interested in classical and quantum error correction, network information and communication theory, and quantum algorithms.\end{biography}

\end{document}